\documentclass[aps,prx,twocolumn,notitlepage,amsmath,amstex,amssymb,citeautoscript,longbibliography]{revtex4-2}

\usepackage{amsmath,amssymb,amsfonts}
\usepackage{graphicx}
\usepackage{tikz}
\usepackage{physics}
\usepackage{amsthm}
\usepackage{hyperref}
\usepackage{algorithm,algpseudocode}
\usepackage{bm}
\usepackage{bbm}
\usepackage[caption=false]{subfig} % use this instead
\newcommand{\ER}{Erd\"{o}s-R\'{e}nyi }
\newcommand{\vgamma}{\bm{\gamma}}
\newcommand{\vbeta}{\bm{\beta}}

\newcommand{\Gbvec}{(\vgamma,\vbeta)}

\newcommand{\Bounds}{[\gamma_{\text{min}},\gamma_{\text{max}})^p\times [\beta_{\text{min}},\beta_{\text{max}})^p}
\newcommand{\Nfev}{n_{\text{fev}}}

\newtheorem{theorem}{Theorem}

\newtheorem{remark}{Remark}

\begin{document}
\title{A Depth-Progressive Initialization Strategy for Quantum Approximate Optimization Algorithm}

\author{Xinwei Lee}
\email{xwlee@cavelab.cs.tsukuba.ac.jp}
\author{Ningyi Xie}
\author{Dongsheng Cai}
\affiliation{University of Tsukuba, Ibaraki Prefecture, Japan}
\author{Yoshiyuki Saito}
\author{Nobuyoshi Asai}
\affiliation{University of Aizu, Fukushima Prefecture, Japan}

\date{\today}

% old abstract
% \begin{abstract}
%     The quantum approximate optimization algorithm (QAOA) is known for its capability and universality in solving combinatorial optimization problems on near-term quantum devices.
% 	The results yielded by QAOA depend strongly on its initial variational parameters $\vgamma$ and $\vbeta$. 
% 	Hence, parameters selection for QAOA becomes an active area of research as bad initialization might deteriorate the quality of the results, 
% 	especially at great circuit depths. 
% 	We first discuss on the patterns of optimal parameters in QAOA in two directions: the angle index and the circuit depth.
% 	Then, we propose a depth-progressive strategy which leverages the patterns in optimal parameters and bounds restriction to aid the initialization. 
% 	We also discuss on the symmetries and periodicity of the Hamiltonians that is used to determine the bounds of the search space.
% 	By forcing the parameters to follow a specific path, we show that they can be initialized near the optimal parameters, increasing the likelihood of the convergence to the optimal.
% 	We then compare this strategy with our previously proposed strategy on solving the Max-cut problem.
% 	We also address the non-optimality in previous parameters, which is seldom discussed in other works, despite its importance in explaining the behavior of 
% 	variational quantum algorithms.
% \end{abstract}

\begin{abstract}
    The quantum approximate optimization algorithm (QAOA) is known for its capability and universality in solving combinatorial optimization problems on near-term quantum devices.
	The results yielded by QAOA depend strongly on its initial variational parameters $\vgamma$ and $\vbeta$. 
	Hence, parameters selection for QAOA becomes an active area of research as bad initialization might deteriorate the quality of the results, 
	especially at great circuit depths. 
	We first discuss the patterns of optimal parameters in QAOA in two directions: the angle index and the circuit depth.
	Then, we discuss the symmetries and periodicity of the expectation that is used to determine the bounds of the search space.
	Based on the patterns in optimal parameters and the bounds restriction, we propose a strategy that predicts the new initial parameters by taking the difference between previous optimal parameters.
	Unlike most other strategies, the strategy we propose does not require multiple trials to ensure success. It only requires one prediction when progressing to the next depth.
	We compare this strategy with our previously proposed strategy and the layerwise strategy on solving the Max-cut problem, in terms of the approximation ratio and the optimization cost.
	We also address the non-optimality in previous parameters, which is seldom discussed in other works, 
	despite its importance in explaining the behavior of variational quantum algorithms.
\end{abstract}

\maketitle

\section{Introduction}
The Quantum Approximate Optimization Algorithm (QAOA) was first introduced by Farhi et al.~\cite{farhi2014quantum} as a quantum-classical hybrid algorithm, which consists of
a quantum circuit with an outer classical optimization loop, to approximate the solution of combinatorial optimization problems.
Since then, many studies are conducted to discuss its quantum advantage and its implementability on near-term Noisy Intermediate Scale Quantum (NISQ)
devices~\cite{Crooks2018PerformanceOT,Guerreschi_2019,farhi2019quantum,Moussa_2020,Marwaha_2021,leo2022sk,akshay2022}.
QAOA is shown to guarantee the approximation ratio $\alpha > 0.6924$ for circuit depth $p=1$ in the Max-cut problem on 3-regular graphs~\cite{farhi2014quantum}.
Further studies have shown a lower bound of $\alpha > 0.7559$ for $p=2$ and $\alpha > 0.7924$ for $p=3$~\cite{Wurtz_2021}.

% talk on angle selection strategies
Parameters selection of QAOA has been an active area of research due to the difficulties that lie within the classical optimization of QAOA, especially the barren plateaus 
problem~\cite{barren_vqa2021,Cerezo_2021,Wang_2021}.
Patterns in optimal parameters of QAOA have constantly been studied and various strategies are proposed to improve the quality of the 
solution~\cite{Grant_2019,zhu2020adaptive,sack2021quantum,multistart,Shaydulin_2021,Alam2020ML,Moussa_2022,amosy2022}.
For instance, it is found that for some classes of graphs, e.g. regular graphs,
the optimal parameters of a smaller graph can be reused as it is on larger graphs to approximate the solution without solving them~\cite{brandao2018fixed}. This characteristic is
defined as the `parameter concentration' in~\cite{akshay2021parameter} and is found in some projectors as well. Recently, the transferability of parameters
is also studied with the discovery of parameters concentration in $d$-regular subgraphs with the same parity (odd or even)~\cite{galda2021transferability}.
These works focused on the characteristics of the optimal parameters of QAOA in the direction of the problem size $n$.

Another direction that is mostly concerned is the QAOA circuit depth $p$. It is discussed that we usually need larger $p$ to solve problems of larger $n$ with higher $\alpha$.
However, as $p$ grows larger, the increased occurrence of local minima makes the optimization difficult. 
If the QAOA parameters are initialized randomly, there is a high chance that they will converge to an undesired local optimum. 
This is shown in our previous work~\cite{lee2021parameters}, and we proposed to use the previous optimal parameters as starting points for the following depths. We found out
that this improves the convergence of the approximation ratio towards better optimal. This implies that there exist some relationship between the previous optimum and the current optimum. This motivates us to study the
relation of the optimal parameters between circuit depths.

In this work, we study the patterns in the optimal parameters of QAOA Max-cut in two directions: the angle index $j$ and the circuit depth $p$.
We name the pattern exhibited by the optima with respect to $j$ as the \emph{adiabatic path}, and the pattern with respect to $p$ as the \emph{non-optimality},
which we explain each of them in detail in Sec.~\ref{sec:pattern}. 
Also, as the expectation function of QAOA Max-cut is highly periodic and symmetric, the landscape it produces will have multiple optima.
Therefore, for the adiabatic path and the non-optimality patterns to be seen explicitly, it is required to restrict the bounds of the parameter search space,
so that the redundant optimal points in the full search space can be removed.

Based on the adiabatic path and the non-optimality, we propose the \emph{bilinear initialization strategy} (or simply \emph{bilinear strategy})
which generates initial parameters at the new depth given the optimal parameters from the previous depths.
This strategy aims to reproduce the optimal patterns at the new depth so that the initial parameters generated will be near to the optimal parameters, 
reducing the likelihood to converge to an undesired optimum. 
Since the previous parameters are used to predict the new parameters, this strategy requires the optimization at every depth up to the desired depth.
However, unlike most other strategies~\cite{leo2020,campos2021} which requires multiple trials to ensure success, the bilinear strategy only requires one trial at each depth.

We then demonstrate the effect of the bilinear strategy on solving the Max-cut problem for 30 non-isomorphic instances composing different classes of graphs,
including the 3-regular, the 4-regular, and the \ER graphs with different edge probabilities. 
We compare the strategy with our previously proposed parameters fixing strategy, and the layerwise strategy, in terms of approximation ratio and optimization cost.

We also study the case where the strategy might fail in odd-regular graphs with wrongly specified bounds.
The result is interesting as it shows that there exists an optimum in the expectation function of odd-regular graphs which does not follow the adiabatic path pattern.
Instead, the $\beta$ parameters oscillate back and forth. We then explain this phenomenon using the symmetry in odd-regular graphs.

\section{QAOA: Background and Notation}
The objective of QAOA is to maximize the expectation of some cost Hamiltonian $H_z$ with respect to the ansatz state $\ket{\psi\Gbvec}$ prepared by the evolution of the alternating operators:
\begin{equation}
	\ket{\psi_p\Gbvec} = \prod^p_{j=1} e^{-i\beta_j H_x}e^{-i\gamma_j H_z}\ket{+}^{\otimes n} \label{eqn:ansatz}.
\end{equation}
where $\vgamma = (\gamma_1, \gamma_2, \ldots, \gamma_p)$ and $\vbeta = (\beta_1, \beta_2, \ldots, \beta_p)$ are the $2p$ variational parameters, with $\vgamma\in[0,2\pi)^p$ and $\vbeta\in[0,\pi)^p$.
$\ket{+}^{\bigotimes n}$ corresponds to $n$ qubits in the ground state of $H_x = \sum_{j=1}^n X_j$, where $X_j$ is the Pauli $X$ operator acting on the $j$-th qubit.

In this paper, we consider the Max-cut problem, which aims to divide a graph into two parts, with the maximum number of edges between them. The Max-cut problem is an NP-complete problem due to
its reducibility to the MAX-2-SAT problem~\cite{Kar72}. The cost Hamiltonian $H_z$ for the Max-cut problem for an unweighted graph $G = (V,E)$ is given as
\begin{equation}
	H_z = \frac{1}{2}\sum_{(j, k)\in E}(\mathbbm{1} - Z_j Z_k),
\end{equation}
where $Z_j$ is the Pauli $Z$ operator acting on the $j$-th qubit.
We define the expectation of $H_z$ with respect to the ansatz state in Eq.~(\ref{eqn:ansatz}):
\begin{equation}
	F_p(\vgamma,\vbeta) \equiv \langle\psi_p(\vgamma,\vbeta)|H_z|\psi_p(\vgamma,\vbeta)\rangle,
	\label{eqn:fp}
\end{equation}
where $p$ is known as the circuit depth of QAOA. Solving the problem with QAOA is equivalent to maximizing Eq.~(\ref{eqn:fp}), with respect to the variational parameters $\vgamma$ and $\vbeta$. This can be done by
a classical optimizer which search for the maximum $F$ and the parameters that maximize it:
\begin{equation}
	(\vgamma^*,\vbeta^*) \equiv \arg\max_{\vgamma,\vbeta} F(\vgamma,\vbeta),
\end{equation}
where the superscript * denotes optimal parameters. We also define the \emph{approximation ratio} $\alpha$ as
\begin{equation}
	\alpha \equiv \frac{F(\vgamma^*,\vbeta^*)}{C_{\text{max}}},
	\label{eqn:alpha}
\end{equation}
where $C_{\text{max}}$ is the maximum cut value for the graph. The approximation ratio is a typical evaluation metric indicating how near the solution given by QAOA is to the true solution, $0\leq\alpha\leq1$, with
the value of 1 nearer to the true solution. 

Throughout the paper, we use the symbol $\phi$ to generally denote either $\gamma$ or $\beta$ in situations where the distinction of both is not required.
Also, we sometimes use $\bm{\Phi}_p$ to denote the entire parameter vector at circuit depth $p$:
\begin{equation}
    \bm{\Phi}_p \equiv \Gbvec_p = (\gamma_1, ..., \gamma_p, \beta_1, ..., \beta_p).
\end{equation}
For a single parameter, we use $\phi_j^p$ to denote the parameter at circuit depth $p$ with index $j$.

The maximum of $F_p$ in the $p$-level search space will approach $C_{\max}$ as $p\rightarrow\infty$, thus the approximation ratio $\alpha$ will approach 1~\cite{farhi2014quantum}.
However, due to the increased occurrence of local maxima in larger $p$, it gets more difficult to find the maximum $F_p$~\cite{leo2020,Guerreschi_2019}. If the parameters were initialized 
randomly, the optimizer is more likely to be trapped in local maxima for larger $p$~\cite{lee2021parameters}. The choice of initial points to the optimizer determines whether the optimizer converges to a global 
maximum. Hence, we would prefer to have ``good'' initial points for the optimizer to converge to the desired maximum.

\begin{figure*}[t]
    \subfloat[]{
        \begin{minipage}[c]{0.32\textwidth}
            \centering
            \includegraphics[width=\textwidth]{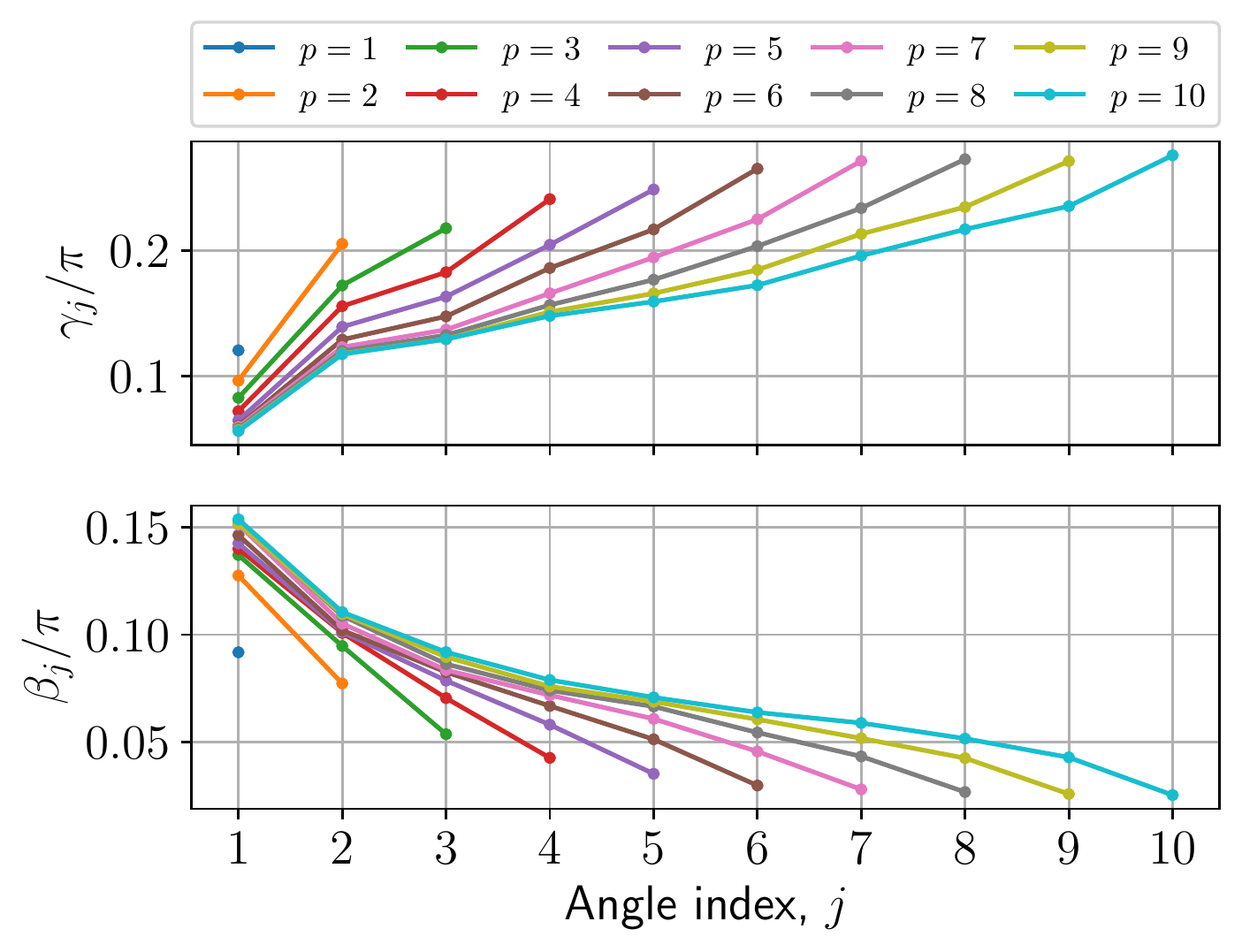}
        \end{minipage}
    }
    \hfill
    \subfloat[]{
        \begin{minipage}[c]{0.32\textwidth}
            \centering
            \includegraphics[width=\textwidth]{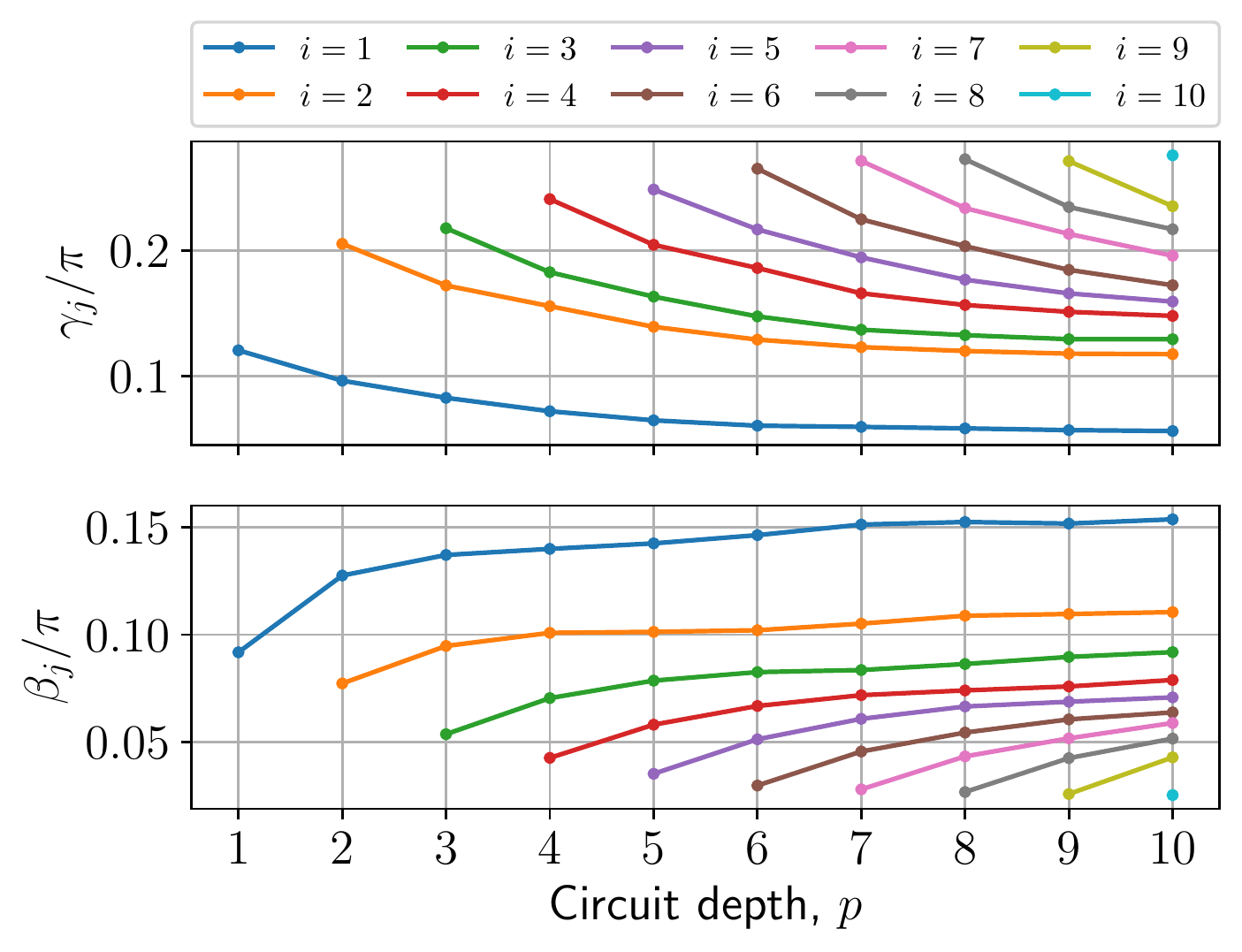}
        \end{minipage}
    }
    \hfill
    \subfloat[]{
        \begin{minipage}[c]{0.32\textwidth}
            \centering
            \includegraphics[width=\textwidth]{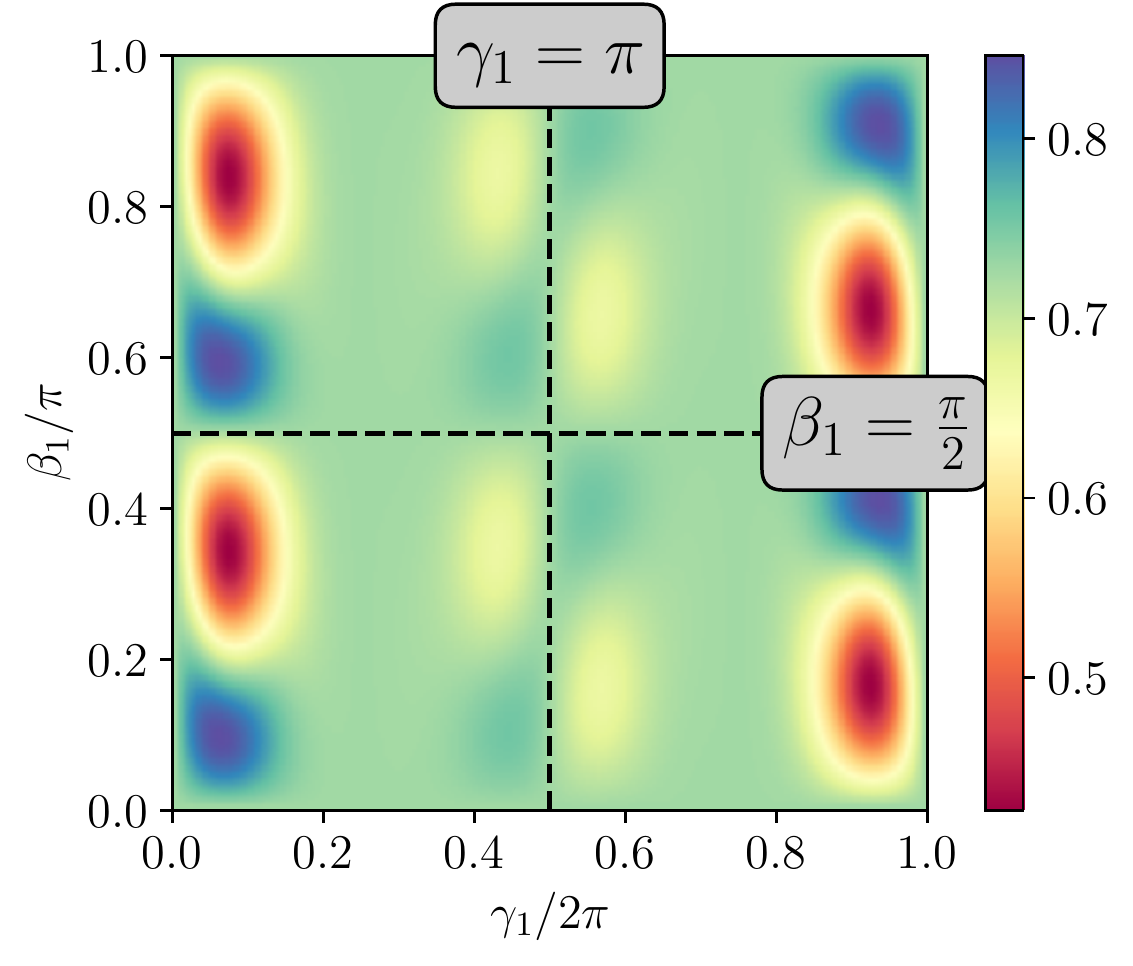}
        \end{minipage}
    }
    \caption{Optimal parameters variation of a 10-node \ER graph with edge probability of 0.7. %
    (a) The variation of the optimal parameters at fixed circuit depth $p$ against the angle index $j$. It shows the adiabatic path of the parameters with increasing $\gamma$ and decreasing $\beta$. %
    (b) The variation of the optimal parameters at fixed angle index $j$ against the circuit depth $p$. It shows the non-optimality of the parameters with decreasing $\gamma$ and increasing $\beta$. %
    (c) The landscape of $p=1$ normalized expectation (i.e. $\alpha$) against $\gamma_1$ and $\beta_1$. The symmetry is shown by the $\gamma_1=\pi$ axis and the periodicity is shown by the $\beta_1=\pi/2$ axis. %
    It can be seen that the landscape in $\gamma_1\in [\pi, 2\pi)$ is the landscape in $\gamma_1\in [0,\pi)$ rotated by $180^{\circ}$, and the landscape just repeats itself beyond $\beta_1=\pi/2$.}
    \label{fig:params-var}
\end{figure*}

\section{Patterns in the optimal parameters of QAOA}\label{sec:pattern}
% patterns in optimal parameters
\subsection{Resemblance to the quantum adiabatic evolution}
It has been repeatedly reported that the optimal parameters resemble the adiabatic quantum computation (AQC) process~\cite{farhi:qaa},
where the mixer Hamiltonian $H_x$ is gradually turned off (decreasing $\beta$) and the cost Hamiltonian $H_z$
is slowly turned on (increasing $\gamma$)~\cite{leo2020,Cook2020TheQA}. 
This comes from the fact that the angles $\gamma$ and $\beta$ is related to the discrete time step of the adiabatic process~\cite{sack2021quantum,farhi:qaa,farhi2014quantum}. 
Consider the time-dependent Hamiltonian $H(t) = (1-t/T)H_x + (t/T)H_z $ going through a simple adiabatic evolution with total run time $T$. Discretizing the evolution gives
\begin{equation}
    e^{-i\int_0^T H(t) dt} \approx \prod_{j=1}^p e^{-iH(j\Delta t)\Delta t},
    \label{eqn:discrete}
\end{equation}
with $t=j\Delta t$. Applying the first order Lie-Suzuki-Trotter decomposition to Eq.~(\ref{eqn:discrete}) gives
\begin{equation}
    (\ref{eqn:discrete}) \approx \prod_{j=1}^p e^{-i(1-j\Delta t/T)H_x\Delta t}e^{-i(j\Delta t/T)H_z\Delta t}.
    \label{eqn:trotter}
\end{equation}
We can then substitute $\gamma_j = (j\Delta t/T)\Delta t$ and $\beta_j = (1-j\Delta t/T)\Delta t$ into Eq.~(\ref{eqn:trotter}), and it leads to the QAOA form in Eq.~(\ref{eqn:ansatz}). Note that the discretization in
Eq.~(\ref{eqn:discrete}) divides the total run time $T$ into $p$ steps, i.e. $\Delta t = T/p$. Hence, we obtain the \emph{parameter-index-depth} relation:
\begin{equation}
    \gamma_j^p = \frac{j}{p}\Delta t;\quad \beta_j^p = \left(1-\frac{j}{p}\right)\Delta t.
    \label{eqn:angle-index-depth}
\end{equation}

It is obvious that $\gamma_j$ increases linearly with $j$ and $\beta_j$ decreases linearly with $j$. 
Previous works~\cite{leo2020,Cook2020TheQA,Crooks2018PerformanceOT,sack2021quantum,Willsch_2020} have shown that the optimal parameters of QAOA tend to follow this linear-like $\emph{adiabatic path}$,
and the pattern becomes nearer to linear as $p$ increases.
Consequently, this pattern is exploited in devising various strategies.
Fig.~\ref{fig:params-var}(a) shows the adiabatic path taken by the optimal parameters at different $p$.
Note that the patterns are not completely linear. 
This might be due to the discretization error and the Trotter error in our process of approximating the continuous evolution,
and it is expected to approach linear as $p\rightarrow \infty$~\cite{Willsch_2020}.

% talk about non-optimality
\subsection{Non-optimality of previous parameters}
Besides the adiabatic path, we also discovered the non-optimality of optimal parameters from previous depths, i.e., the optimal parameters for $p$ are not optimal for $p+1$. 
The optimal parameters are shifted by a little as the depth increases, as observed in Fig.~\ref{fig:params-var}(b).
% \begin{conjecture}[Non-optimality of previous parameters]
%     Before saturation, the optimal parameters $\bm{\Phi^*} = \Optvec$ from circuit depths $p$ and $p+k$ obey
%     \begin{equation}
%         \forall j \leq p,\quad \phi_j^p \neq \phi_j^{p+k}, %\ \text{and}\quad \beta_i^p \neq \beta_i^{p+1}
%     \end{equation}
%     for $k\in\mathbb{N}$. $\phi_j^p$ is the optimal parameter at circuit depth $p$ with index $j$.
%     \label{conj:non-op}
% \end{conjecture}
% Saturation means the approximation ratio has reached its maximum value, and increasing $p$ will not increase the approximation ratio: $\alpha_p = \alpha_{p+1} = 1$.
% Conjecture~\ref{conj:non-op} is based on our observation on the numerical results for many Max-cut instances.
This phenomenon can also be inferred from Eq.~(\ref{eqn:angle-index-depth}). As $p$ increases, at the same index $j$, $\gamma_j$ will decrease and $\beta_j$ will increase.
Also, we noticed that as $p$ gets larger, the parameters at smaller indices have less changes compared to those with larger indices, e.g.,
it can be observed, in Fig.~\ref{fig:params-var}(b), that $|\phi_8^{10} - \phi_8^9|$ (rightmost two points of the gray line) is greater than $|\phi_1^{10} - \phi_1^9|$ (rightmost two of the blue line).
We emphasize that the non-optimality is just the counterpart of the adiabatic path as they can be explained with the same relation, but it is seldom discussed in previous works.
The patterns in the optimal parameters seem to be inherited from the time steps in the discrete adiabatic evolution.

The results of layerwise training of QAOA also implies this non-optimality~\cite{campos2021}. 
Layerwise training is an optimization strategy in which only the parameters of the current layer are optimized, the rest of the parameters are taken from the previous optimal parameters. 
The layerwise training has a relatively low training cost in exchange for a lower approximation ratio, as it suffers from premature saturation (saturation before the approximation ratio reaches 1).
In~\cite{campos2021}, the authors discussed the premature saturation at $p=n$ for the rank-1 projector Hamiltonian $H_z = \ketbra{0^n}$, where $n$ is the number of qubits.
Since the previous parameters in layerwise training are held constant and not allowed to move throughout the optimization, it will not reach the global minimum because of the non-optimality.
% \begin{definition}[Saturation]
% The approximation (or expectation) reaches its maximum possible value. After saturation, progressing the depth will not increase the approximation ratio,
% \begin{equation}
%     \alpha_p = \alpha_{p+1} = 1.
% \end{equation}
% \end{definition}

\subsection{Bounded optimization of QAOA}\label{subsec:bounded}
% bounded optimization
The adiabatic path and the non-optimality shows that the optimal parameters exhibit some trends related to AQC.
However, in the QAOA parameter space, not only the adiabatic path leads to the solution of the problem. There are redundant optimal points in the search space which their patterns
do not follow the adiabatic path. This leads to the optimization of the bounded parameter space where no redundancy exists in it. Therefore, the properties of the problem and its parameter space need to be 
studied beforehand to ensure only one optimum exists in the parameter space.

For instance, the bounds of the unweighted Max-cut problem were originally taken as $\vgamma\in [0,2\pi)^p$ and $\vbeta\in [0, \pi)^p$ because of their periodicity~\cite{farhi2014quantum}.
However, it is further discussed in~\cite{leo2020} that the operator $e^{-i(\pi/2)H_x} = X^{\otimes n}$ commutes through the operators in Eq.~(\ref{eqn:ansatz}), 
and due to the symmetry of the solutions, the period of $\beta$ 
becomes $\pi/2$. Also, QAOA has an time-reversal symmetry: 
\begin{equation}
    F_p(\vgamma, \vbeta) = F_p(-\vgamma, -\vbeta) = F_p\left(2\pi-\vgamma, \frac{\pi}{2}-\vbeta\right).
    \label{eqn:arsym}
\end{equation}
The second equality is due to the fact that $\gamma$ has a period of $2\pi$ and $\beta$ has a period of $\pi/2$. 
From Eq.~(\ref{eqn:arsym}), one would expect that the landscape of $F_p$ beyond $\vgamma = \pi$ is the
image of rotation by $180^{\circ}$ of the landscape within $\vgamma=\pi$ (correspond to the reflection of both $\vgamma=\pi$ and $\vbeta=\pi/4$). 
Therefore, in general, the optimization can be done in the bounds $\bm{\Phi}_p \in [0,\pi)^p\times [0,\pi/2)^p$ due to the redundancies in the landscape,
i.e., one part of the landscape being the image of another. 
Fig.~\ref{fig:params-var}(c) shows the visualization of the $p=1$ expectation landscape for a 10-node \ER graph.
It is observed that the maximum point (colored in blue) is redundant beyond $\gamma_1 = \pi$ and $\beta_1 = \pi/2$.
Moreover, in regular graphs, there are symmetries in $e^{-i\pi H_z} = Z^{\otimes n}$ for odd degree regular graphs, and $e^{-i\pi H_z} = \mathbbm{1}$ for even degree regular graphs. 
Thus, the optimization bounds can be further restricted to $[0,\pi/2)^p\times [0,\pi/2)^p$ for unweighted regular graphs.
The details for the periodicity and symmetries are mainly discussed in~\cite{leo2020,Lotshaw_2021,pt-weighted}, and we include the derivations in Appendix~\ref{sec:prop-qaoa}.
In Table~\ref{tab:good-bounds}, we summarized the suitable bounds for $\vgamma$ and $\vbeta$ for different types of graphs to avoid redundancies in the expectation landscape.

\begin{table}[h]
    \centering
    \caption{The suitable optimization bounds for different types of graphs for the QAOA of unweighted Max-cut to avoid redundant optimal points.}
    \begin{tabular}{ccc}
        \hline 
        & $\vgamma$ bounds & $\vbeta$ bounds \\
        \hline
        Regular graphs & $[0,\pi/2)$ & $[0,\pi/2)$ \\
        All other graphs & $[0,\pi)$ & $[0,\pi/2)$ \\
        \hline
    \end{tabular}
    \label{tab:good-bounds}
\end{table}

\section{Bilinear strategy}
Using the properties discussed in Sec.~\ref{sec:pattern}, we devise a strategy that is depth-progressive, i.e.,
the optimization is done depth-by-depth up to the desired depth $p$.
We utilize the fact that in the bounded search space, the optimal parameters undergo smooth changes as shown in
the patterns of the adiabatic path and the non-optimality.
Our strategy tries to reproduce the adiabatic path and the non-optimality patterns so that we can generate initial parameters that are near to the optimal. 
Therefore, we use the difference in previous optimal parameters to predict the initial points for the new parameters, i.e., the parameters for the next depth.
Following the adiabatic path, we can predict $\phi_{j+1}^p$ using $\Delta_{j,j-1}^p \equiv \phi_j^p - \phi_{j-1}^p$.
For the non-optimality, we can predict $\phi_{j}^{p+1}$ using $\Delta_j^{p,p-1} \equiv \phi_j^p - \phi_j^{p-1}$.
We define $\Delta_{i,j}$ as the difference between the parameters $\phi_i$ and $\phi_j$, and this works the same way for the superscript.
We call this the \emph{bilinear strategy} as it involves the linear differences of two directions: $j$ and $p$.

We explain the mechanism of our strategy. First, we can use any exhaustion method to find the optima for $p=1$
and $p=2$ within the specified bounds $\bm{\Phi}_p\in \Bounds$. 
This is to establish the base for our strategy, where we can take the difference between two sets of optimal parameters. 
The bounds are chosen such that there are no redundant optimal in the search space, as mentioned in Sec.~\ref{subsec:bounded}, so that we can capture the pattern.
We start applying the strategy from $p=3$.
The parameters with indices up to $j=p-2$ are extrapolated using the pattern of non-optimality:
\begin{equation}
    \begin{split}
        \forall j\leq p-2,\quad \phi_j^p & = \phi_j^{p-1} + \Delta_j^{p-1,p-2} \\
        & = 2\phi_j^{p-1} - \phi_j^{p-2}.
    \end{split}
    \label{eqn:jp-2}
\end{equation}
The current parameter $\phi_j^p$ is extended from the previous parameter $\phi_j^{p-1}$
by adding the difference between the previous two parameters $\Delta_j^{p-1,p-2} = \phi_j^{p-1}-\phi_j^{p-2}$.
Note that $\Delta$ can be either positive or negative, which determines the direction of the extrapolation.
This agrees with the monotonous change of the optimal parameters. 
For the parameters with index $j=p-1$, we want to use a relation similar to Eq.~(\ref{eqn:jp-2}).
However, the parameter $\phi_{p-1}^{p-2}$ does not exist, so we take the difference from the previous index $j=p-2$ instead:
\begin{equation}
    \phi_{p-1}^p = \phi_{p-1}^{p-1} + \Delta_{p-2}^{p-1,p-2}.
    \label{eqn:jp-1}
\end{equation}
For the newly added parameter $j=p$, it is predicted using the adiabatic path pattern:
\begin{equation}
    \begin{split}
        \phi_p^p & = \phi_{p-1}^p + \Delta_{p-1,p-2}^p \\
        & = 2\phi_{p-1}^p - \phi_{p-2}^p.
    \end{split}
    \label{eqn:jp}
\end{equation}
If the parameters produced in Eq.~(\ref{eqn:jp-2}) -- (\ref{eqn:jp}) are out of the bounds specified,
we will take the boundary value of $\phi_{\text{min}}$ or $\phi_{\text{max}}$ (whichever is nearer) to replace them.
After the process, the initial parameters $\bm{\Phi}_p = (\gamma_1, ..., \gamma_p, \beta_1, ..., \beta_p)$ for $p$ will be obtained, and it is optimized to find the optimal at $p$. 
This entire process is summarized in Algorithm~\ref{alg:bl}. A visualization diagram of the strategy is also shown in Fig.~\ref{fig:bl}.

% algorithm here
\begin{algorithm}[H]
    \caption{Bilinear initialization}
    \begin{algorithmic}[1]
        \State \textbf{Input:} $\bm{\Phi}_1^*$ and $\bm{\Phi}_2^*$. $\bm{\Phi}_p\in \Bounds$.
        \For{$p:=3...q$}
            \State Build the initial parameters $\bm{\Phi_p}$:
            \For{$j:=1...p$}
                \If{$j\leq p-2$}
                    \State $\phi_j^p \leftarrow \phi_j^{p-1} + \Delta_j^{p-1,p-2}$
                \ElsIf{$j=p-1$}
                    \State $\phi_j^p \leftarrow \phi_j^{p-1} + \Delta_{j-1}^{p-1,p-2}$
                \ElsIf{$j=p$}
                    \State $\phi_j^p \leftarrow \phi_{j-1}^p + \Delta_{j-1,j-2}^p$
                \EndIf
            \EndFor
            \State Initialize QAOA with $\bm{\Phi}_p$ and perform bounded optimization.
        \EndFor
        \State \textbf{Output:} $\bm{\Phi}_p^*$ and $F_p(\bm{\Phi}_p^*)$ for each $p$ up to $q$.
    \end{algorithmic}
    \label{alg:bl}
\end{algorithm}

\begin{figure}
    \centering
    \includegraphics[width=0.45\textwidth]{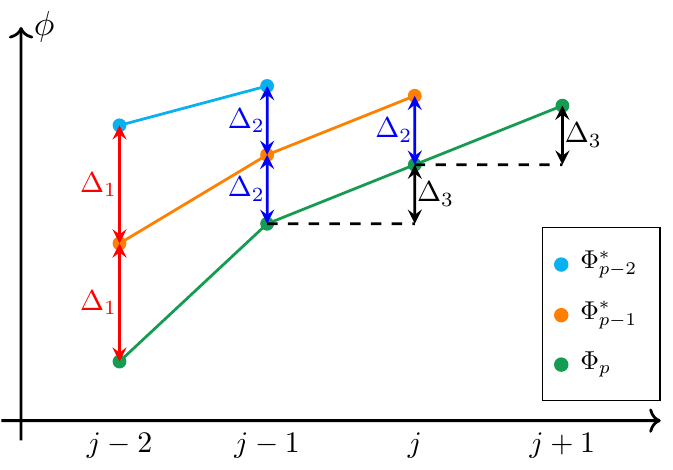}
    \caption{Visualization of the bilinear strategy. $\Delta_1$, $\Delta_2$, and $\Delta_3$ correspond to the values calculated in Eq.~(\ref{eqn:jp-2}), (\ref{eqn:jp-1}), and (\ref{eqn:jp}) respectively. %
    $\Delta_1$ and $\Delta_2$ represent the change due to the non-optimality. $\Delta_3$ represents the change due to the adiabatic path. $\bm{\Phi}_p$ is the new initial parameters extrapolated from %
    $\bm{\Phi}_{p-1}^*$ and $\bm{\Phi}_{p-2}^*$.}
    \label{fig:bl}
\end{figure}

\section{Results}
We apply the bilinear strategy on solving the Max-cut problem for regular graphs and \ER graphs. 
The performance of the strategy is evaluated on 30 non-isomorphic instances of different classes of graphs up to the number of nodes $n=20$, 
which include the 3-regular, 4-regular graphs, and \ER graphs with different edge probabilities.
For the regular graphs, we optimize the parameters within the bound $[0,\pi/2)^p\times [0,\pi/2)^p$, 
whereas the \ER graphs are optimized within $[0,\pi)^p\times [0,\pi/2)^p$.
Here, we only show the results for 4 of the instances in Fig.~\ref{fig:results}, but the trends discussed also apply to all other instances unless particularly stated.

We compare the approximation ratio $\alpha$ obtained from the bilinear strategy with our previously proposed parameters fixing strategy~\cite{lee2021parameters}.
From Fig.~\ref{fig:results}(a)-(d), it can be observed that the $\alpha$ produced by the bilinear strategy traces the optimal $\alpha$ (found by parameters fixing) with minimal error.
The results of the layerwise strategy is also plotted.
Besides the projectors done in the previous work~\cite{campos2021}, we found out that for the Max-cut Hamiltonian, $\alpha$ also saturates at a certain $p$ due to the non-optimality of the parameters.

In Fig.~\ref{fig:results}(e)-(h), we compare the number of function evaluations $\Nfev$ before the convergence of the Limited-memory BFGS Bounded (L-BFGS-B)~\cite{l-bfgs-b} optimizer for the strategies.
It is the number of function calls to the quantum circuit to compute the expectation in Eq.~(\ref{eqn:fp}), and less $\Nfev$ usually means less quantum and classical resources used.
For parameters fixing and layerwise which need multiple trials to ensure success, we consider the total $\Nfev$ for 20 trials.
The results show that the $\Nfev$ required by the bilinear initial points are always less than that of the parameters fixing by an order of $10^2$ to $10^3$ for $p\geq 3$.
This clearly shows the advantage of the bilinear strategy on the optimization cost as only a single trial is required.
For $p=1$ and $p=2$, the $\Nfev$'s are the same as we used parameters fixing to search for the optima. 
As for layerwise, the $\Nfev$'s are relatively small, as only 2 parameters are optimized for each $p$.
It is observed that for small depths up to $p=6$, even the bilinear strategy cost less than layerwise, and the cost grows with $p$ as the number of optimization variables increases.

\begin{figure*}[t]
    \includegraphics[width=\textwidth]{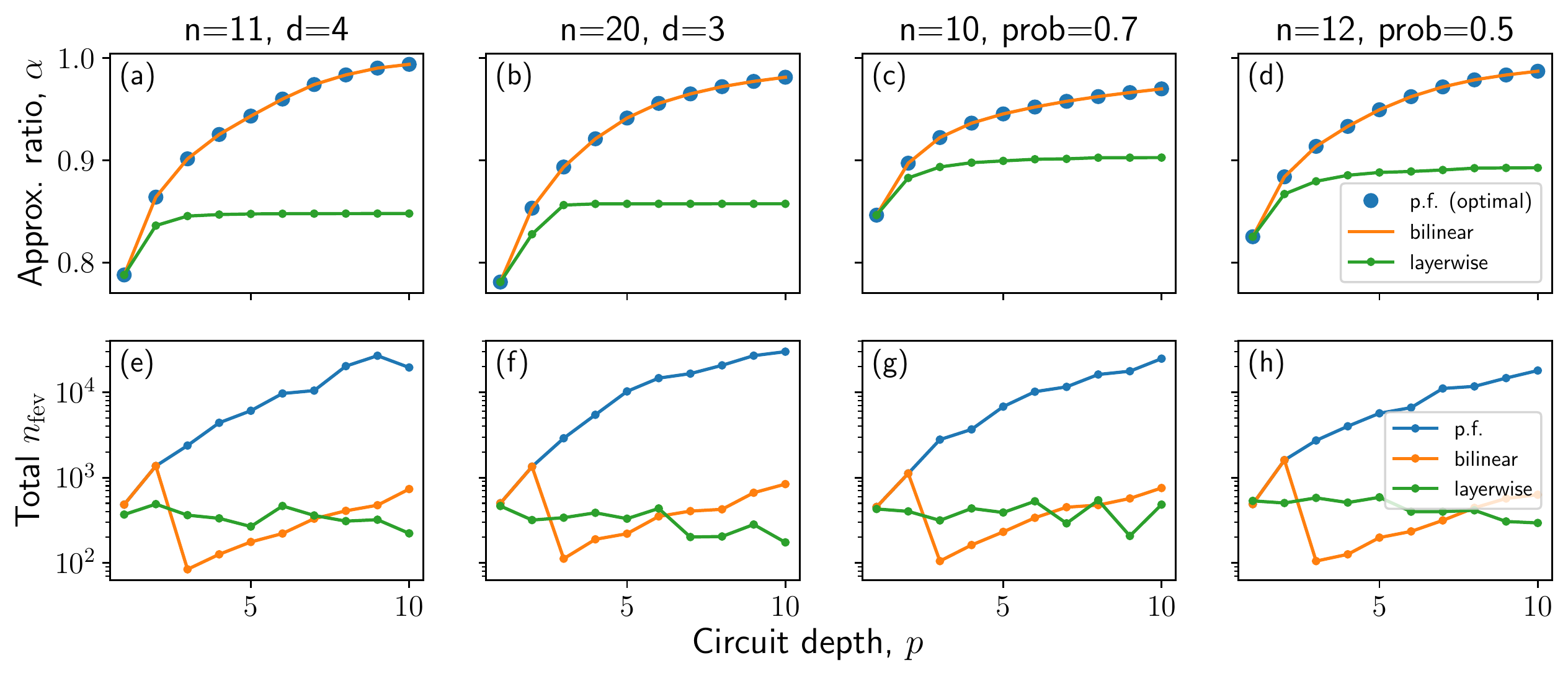}
    \caption{Comparison of the results for parameters fixing, layerwise, and the newly proposed bilinear strategy. %
    $n$ is the number of nodes/vertices of the graph, $d$ is the degree for regular graphs, `prob' is the edge probability for \ER graphs. %
    (a)-(d) show the changes in the approximation ratio $\alpha$ against $p$. (e)-(h) show the $\Nfev$ required before convergence at different $p$'s for the L-BFGS-B optimizer (log scale). %
    For parameters fixing and layerwise, the $\Nfev$ is the total of 20 trials.}
    \label{fig:results}
\end{figure*}

On the other hand, we also consider the case where the bilinear strategy fails, where the initial points do not follow the monotonous trend of the adiabatic path. 
One of the examples is the odd-regular graphs. We mentioned in Section~\ref{subsec:bounded} that one should take the bound $[0,\pi/2)^p\times [0,\pi/2)^p$ for regular graphs to avoid redundancies.
However, if one tries to take the bound $[0,\pi)^p\times [0,\pi/2)^p$, which is considered the general bound for unweighted Max-cut, 
one has chance to fall into the starting point in $\gamma_1\in [\pi/2, \pi)$ for $p=1$ (shown in Fig.~\ref{fig:non-adiabatic-start}(a)). 
In this case, the optimal parameters will not follow the adiabatic path as shown in Fig.~\ref{fig:params-var}(a).
Fig.~\ref{fig:non-adiabatic-start}(b) shows that for this non-adiabatic starting point, the optimal $\beta$'s oscillate back and forth instead.
In fact, this point is symmetric to the adiabatic start $\gamma_1\in [0, \pi/2)$. 
We explain this odd-regular symmetry, including the $\beta$ oscillation in Appendix~\ref{sec:non-adiabatic}.

Fig.~\ref{fig:non-adiabatic-start}(c) shows the result of bilinear strategy with a non-adiabatic start for a 10-node 3-regular graph. 
The $\alpha$ produced by the bilinear strategy traces the optimal until $p=7$, where it deviates after $p=8$.
This shows that the bilinear strategy is still effective to some extent, even for non-adiabatic starts.

\begin{figure*}
    \subfloat[]{
        \begin{minipage}[c]{0.32\textwidth}
            \centering
            \includegraphics[width=\textwidth]{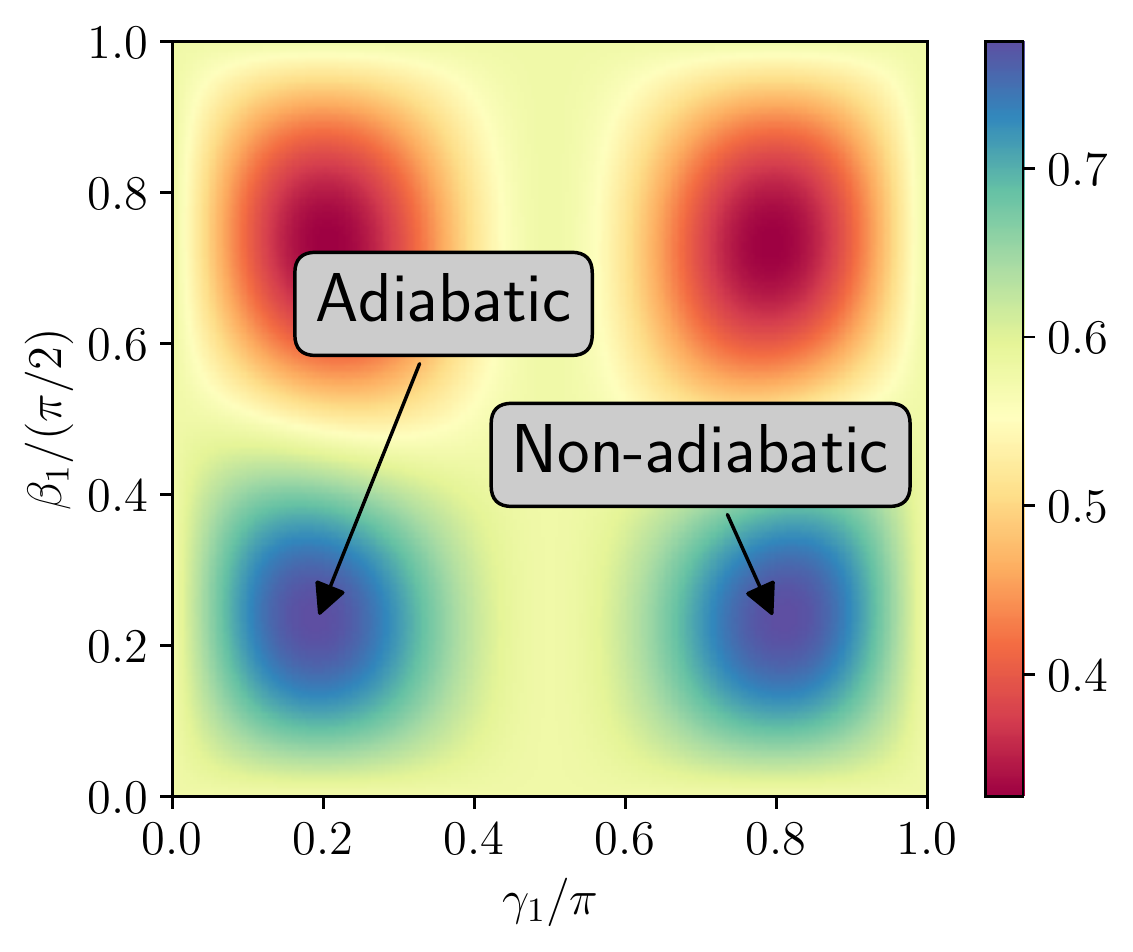}
        \end{minipage}
    }
    \hfill
    \subfloat[]{
        \begin{minipage}[c]{0.32\textwidth}
            \centering
            \includegraphics[width=\textwidth]{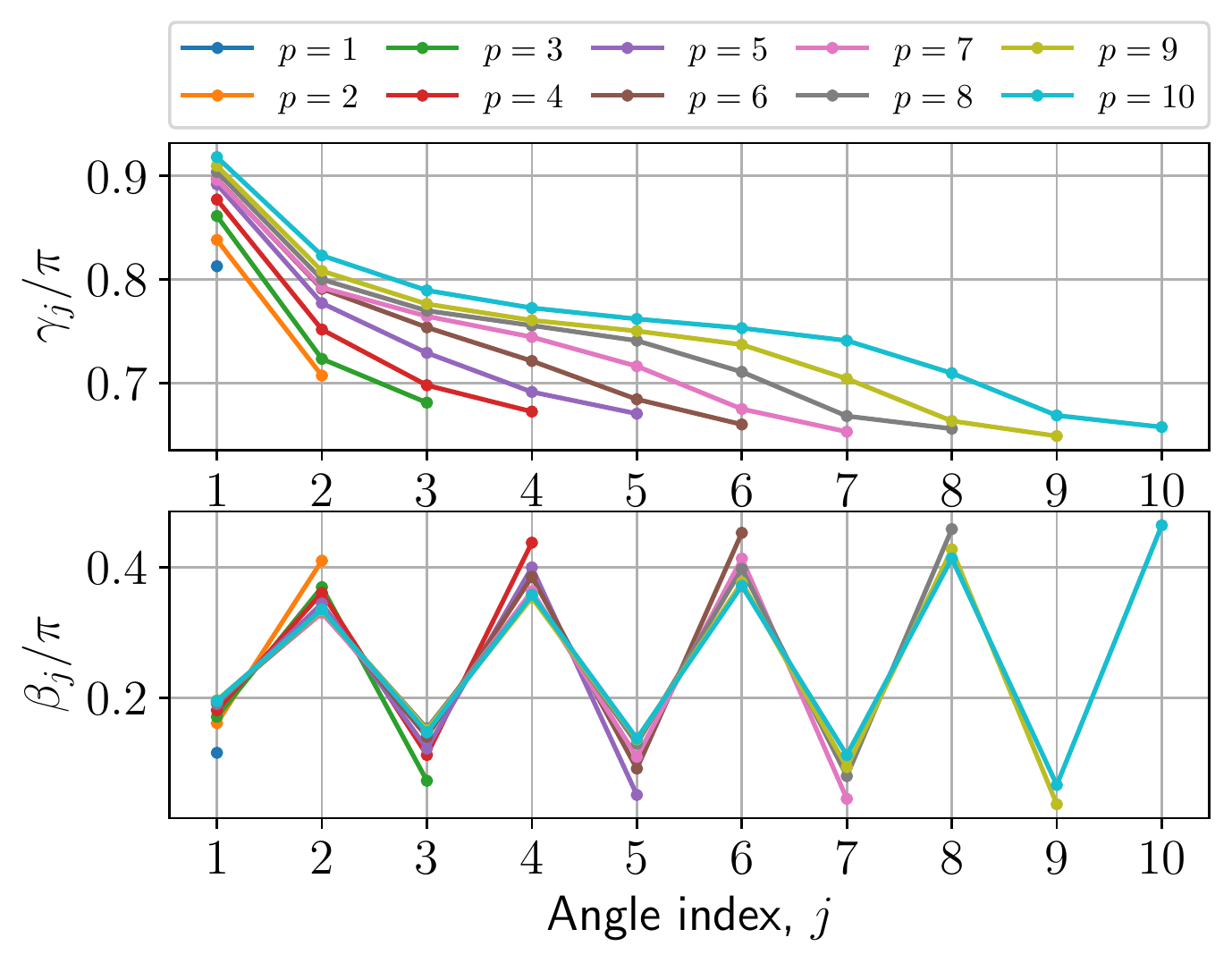}
        \end{minipage}
    }
    \hfill
    \subfloat[]{
        \begin{minipage}[c]{0.32\textwidth}
            \centering
            \includegraphics[width=\textwidth]{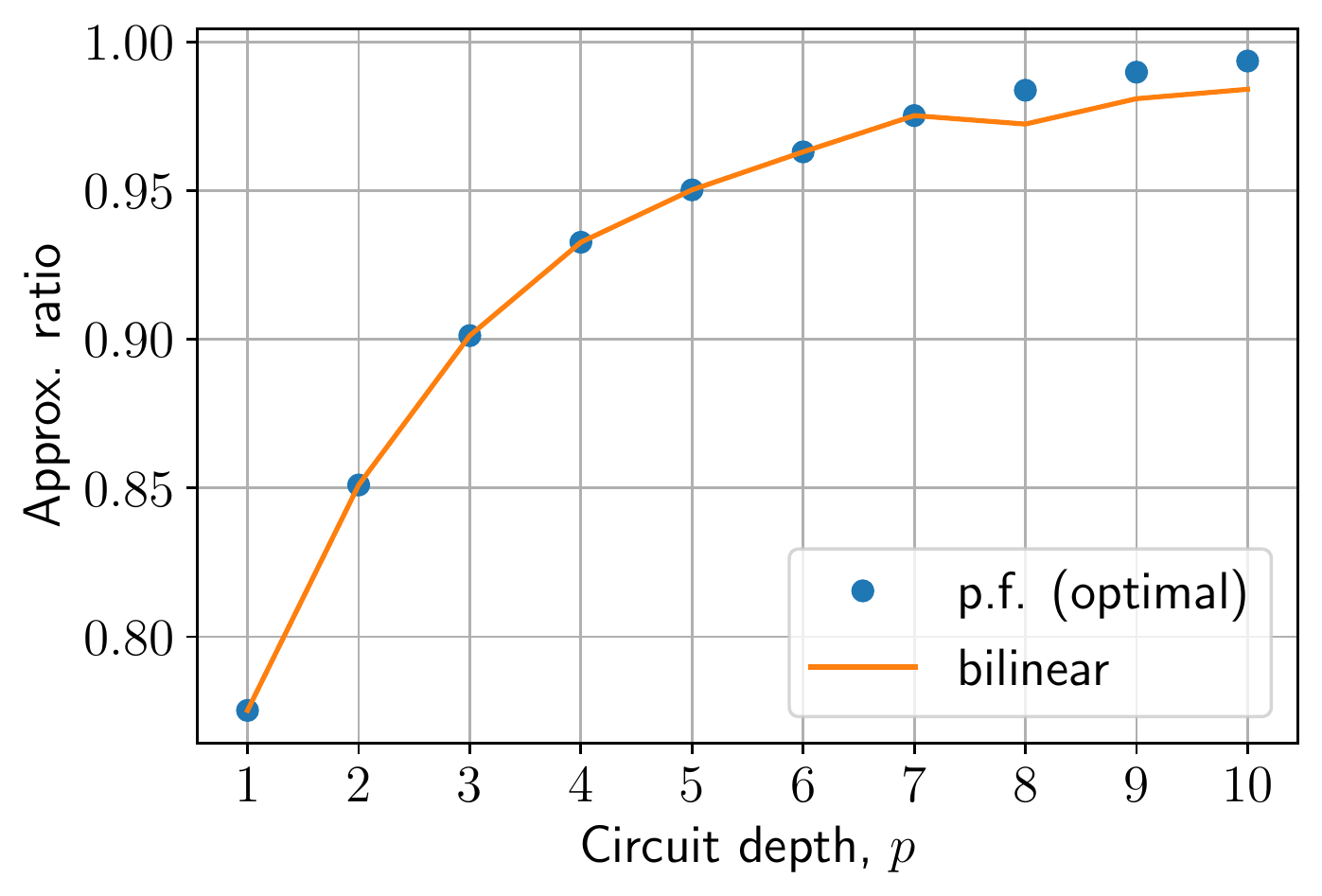}
        \end{minipage}
    }
    \caption{(a) $p=1$ normalized expectation (i.e. $\alpha$) landscape for a 10-node 3-regular graph showing multiple maxima in $\gamma_1\in [0, \pi)$. %
    When used as a starting point in the bilinear strategy, the maximum on the left follows the adiabatic path, whereas the maximum on the right does not follow the adiabatic path. %
    (b) The variation of the optimal parameters with a non-adiabatic start. Unlike the adiabatic start, the $\beta$'s oscillate back and forth. %
    (c) The effect of the bilinear strategy under the non-adiabatic start. }
    \label{fig:non-adiabatic-start}
\end{figure*}

\section{Conclusion and Outlook}
To conclude, we have studied the patterns in the optimal parameters of QAOA for the unweighted Max-cut problem in two directions, namely the angle index $j$ and the circuit depth $p$.
We call the variation against $j$ and $p$ the \emph{adiabatic path} and the \emph{non-optimality} respectively. 
By leveraging these properties, we devise the depth-progressive bilinear strategy, in which the optimization is done for each depth until the desired depth. 
The bilinear strategy utilizes the optimal parameters from the previous two depths, $\bm{\Phi}_{p-2}^*$ and $\bm{\Phi}_{p-1}^*$, to initialize the parameters for the current depth $\bm{\Phi}_p$.

We have demonstrated the effectiveness of the bilinear strategy by comparing it with the parameters fixing strategy~\cite{lee2021parameters} and the layerwise~\cite{campos2021}
strategy on 30 non-isomorphic random regular and \ER graphs.
The results show bilinear is able to trace the optimal approximation ratio $\alpha$ found by parameters fixing.
Whilst, we have also observed the premature saturation occurring in the Max-cut Hamiltonian for layerwise. 
It is also found out that the number of function evaluations $\Nfev$ of bilinear is less than that of parameters fixing due to its single prediction.

The bilinear strategy is advantageous against most other strategies~\cite{leo2020,campos2021} that usually require multiple trials to ensure success, including the parameters fixing strategy.
It only requires the optimization of one set of initial parameters at each circuit depth.
The bilinear strategy also requires the knowledge of the bounds for the optimization to avoid redundancies in the search space, hence ensuring success.
However, we have considered the case where it fails when initialized from a ``non-adiabatic'' point.
Numerically, for a particular 3-regular graph, it is still capable of tracing the optimal until circuit depth $p=7$.

We suggest some potential work that can be done in the future.
Since the new prediction is extrapolated from the change of the optimal parameters, we can increase the depth-step of the bilinear strategy for less optimization cost.
In this work, we have shown using the depth-step of 1. One can, for example, increase the depth-step to 2 ($p=2, 4, 6, ...$) while progressing to larger circuit depths.
Although not tested, the bilinear strategy is expected to perform on different kinds of problems (different $H_z$) in which their optimal parameters follow the adiabatic path and
non-optimality, which is believed to be true for QAOA. 
This is also a good future work to explore.

\appendix

\section{Properties of QAOA Max-cut}\label{sec:prop-qaoa}
The QAOA for the Max-cut problem is highly periodic and symmetric. This is addressed in several works~\cite{leo2020,Lotshaw_2021,pt-weighted}.
In this section, we derive the properties that help us to avoid global optima redundancies and to determine the bounds for the optimization.

\begin{theorem}[Angle-reversal symmetry of QAOA]
    The expectation of QAOA  stays the same when its angles (parameters) are negated.
    \begin{equation}
        F_p(\vgamma, \vbeta) = F_p(-\vgamma, -\vbeta),
        \label{eqn:angle-reversal}
    \end{equation}
    for any circuit depth $p$. This is true for any Hermitian mixer and problem Hamiltonian $H_x$ and $H_z$.
\end{theorem}

\begin{proof}
    We use the fact that the expectation $F(\vgamma, \vbeta)$ is real, so its complex conjugate is just itself:
    $F(\vgamma, \vbeta) = \overline{F(\vgamma, \vbeta)}$. 
    Hence, for any Hermitian matrices $H_x$ and $H_z$,
    \begin{equation}
        F_p(\vgamma, \vbeta) = \overline{F_p(\vgamma, \vbeta)}
    \end{equation}
    \begin{multline}
        \bra{+}^{\otimes n}e^{i\beta_1H_x}e^{i\gamma_1H_z}\cdots H_z
            \cdots e^{-i\beta_1H_x}e^{-i\gamma_1H_z}\ket{+}^{\otimes n} \\
            = \bra{+}^{\otimes n}e^{-i\beta_1H_x}e^{-i\gamma_1H_z}\cdots H_z
            \cdots e^{i\beta_1H_x}e^{i\gamma_1H_z}\ket{+}^{\otimes n}
    \end{multline}
    \begin{equation}
        F_p(\vgamma, \vbeta) = F_p(-\vgamma, -\vbeta).
    \end{equation}
\end{proof}

\begin{theorem}[General periodicity and symmetry for unweighted Max-cut]
    The expectation function of the unweighted Max-cut problem for any graph has a period of $2\pi$ w.r.t. the parameter(s) $\vgamma$, 
    and a period of $\pi/2$ w.r.t. the parameter(s) $\vbeta$.
    \begin{align}
        F_p(\vgamma, \vbeta) & = F_p(\vgamma + 2\pi, \vbeta) \label{eqn:period-gamma} \\
        & = F_p\left(\vgamma, \vbeta + \frac{\pi}{2}\right) \label{eqn:period-beta}\\
        & = F_p\left(\vgamma + 2\pi, \vbeta + \frac{\pi}{2}\right), \label{eqn:period-gamma-beta}
    \end{align}
    where $\bm{\phi} + c$ means a shift of every element in the parameter vector $\bm{\phi}$ by a scalar $c$,
    $\bm{\phi} + c = (\phi_1+c, \phi_2+c, ..., \phi_p+c)$.
    
    Combining the angle-reversal and the periodicity creates a symmetry on the expectation function:
    \begin{equation}
        F_p(\vgamma, \vbeta) = F_p\left(2\pi-\vgamma, \frac{\pi}{2}-\vbeta\right).
        \label{eqn:general-symmetry}
    \end{equation}
    \label{theorem:period}
\end{theorem}

\begin{proof}
    It is known that $e^{-i(2\pi) H_z} = \mathbbm{1}$.
    Consider the ansatz with $\vgamma$ shifted by $2\pi$:
    \begin{align}
        \ket{\psi_p(\vgamma+2\pi, \vbeta)} & = \prod_{j=1}^p e^{-i\beta_jH_x}e^{-i\gamma_jH_z}e^{-i(2\pi) H_z}\ket{+}^{\otimes n}\\
        & = \prod_{j=1}^p e^{-i\beta_jH_x} e^{-i\gamma_jH_z} \mathbbm{1} \ket{+}^{\otimes n} \\
        & = \prod_{j=1}^p e^{-i\beta_jH_x} e^{-i\gamma_jH_z} \ket{+}^{\otimes n} \\
        & = \ket{\psi_p(\vgamma, \vbeta)}.
    \end{align}
    Since the ansatz stays the same under the shift, so is the expectation, hence proving Eq.~(\ref{eqn:period-gamma}).
    
    It is known that $e^{i(\pi/2)H_x} = X^{\otimes n}$ commutes with the QAOA operators 
    $e^{-i\beta H_x}$ and $e^{-i\gamma H_z}$. 
    The former commutation is due to the rotation of the same Pauli. The latter is due to the symmetry of the eigenstates of $H_z$,
    e.g., the eigenvalue (or cut-value) of $\ket{0110}$ is equal to the eigenvalue of $\ket{1001}$. The eigenvalues are invariant under the bit-flip operation $X^{\otimes n}$.
    Consider the ansatz with $\vbeta$ shifted by $\pi/2$:
    \begin{align}
        \ket{\psi_p\left(\vgamma, \vbeta+\frac{\pi}{2}\right)} & = \prod_{j=1}^p e^{-i\beta_jH_x} e^{-i(\pi/2)H_x} e^{-i\gamma_jH_z} \ket{+}^{\otimes n}\\
        & = \prod_{j=1}^p e^{-i\beta_jH_x} X^{\otimes n} e^{-i\gamma_jH_z} \ket{+}^{\otimes n}.
    \end{align}
    Since $X^{\otimes n}$ commutes through the operators, we can move all of them to the rightmost before the initial state $\ket{+}^{\otimes n}$.
    Note that since $\ket{+}^{\otimes n}$ is the eigenstate of $X^{\otimes n}$, it will not change the initial state. 
    Therefore, we have
    \begin{align}
        \ket{\psi_p\left(\vgamma, \vbeta+\frac{\pi}{2}\right)} & = \prod_{j=1}^p e^{-i\beta_jH_x} e^{-i\gamma_jH_z} X^{\otimes n} \ket{+}^{\otimes n}\\
        & = \prod_{j=1}^p e^{-i\beta_jH_x} e^{-i\gamma_jH_z} \ket{+}^{\otimes n} \\
        & = \ket{\psi_p(\vgamma, \vbeta)}.
    \end{align}
    Hence proving Eq.~(\ref{eqn:period-beta}).
    
    Combining Eq.~(\ref{eqn:angle-reversal}) and (\ref{eqn:period-gamma-beta}), Eq.~(\ref{eqn:general-symmetry}) can be derived.
    \begin{equation}
        F_p(\vgamma, \vbeta) = F_p(-\vgamma, -\vbeta) = F_p \left(2\pi-\vgamma, \frac{\pi}{2}-\vbeta\right),
    \end{equation}
    showing the expectation function reflects over the axis $\vgamma = \pi$, and then the axis $\vbeta = \pi/4$. 
    This is equivalent to a $180^{\circ}$ rotation about the point $(\vec{\pi}, \vec{\pi}/4)$, where $\vec{\pi} = (\pi, \pi, ..., \pi)$.
\end{proof}

\begin{remark}
The proof of Theorem~\ref{theorem:period} shows that the periodicity is not only true for the shift for the parameter vector $\bm{\phi}$, 
but also true for any arbitrary number of $\phi_j$ shift by its corresponding period, as any number of the operator $\mathbbm{1}$ (or $X^{\otimes n}$)
will still be canceled out.
\label{remark:any-shift}
\end{remark}

\begin{theorem}[Periodcity and symmetry for even-regular graphs]
    For the Max-cut of even-regular graphs, the expectation has a period of $\pi$ w.r.t. the parameter(s) $\vgamma$, 
    which is shorter than the general period.
    \begin{align}
        F_p(\vgamma, \vbeta) & = F_p(\vgamma+\pi, \vbeta) \\
        & = F_p\left(\vgamma+\pi, \vbeta+\frac{\pi}{2}\right). \label{eqn:even-reg-period}
    \end{align}
    Due to the angle-reversal, the expectation of even-regular graphs has the symmetry
    \begin{equation}
        F_p(\vgamma, \vbeta) = F_p\left(\pi-\vgamma, \frac{\pi}{2}-\vbeta\right).
        \label{eqn:even-reg-symmetry}
    \end{equation}
\end{theorem}

\begin{proof}
    For the $H_z$ of even-regular graphs, $e^{-i\pi H_z} = \mathbbm{1}$.
    Consider the ansatz with $\vgamma$ shifted by $\pi$:
    \begin{align}
        \ket{\psi_p(\vgamma+\pi, \vbeta)} & = \prod_{j=1}^p e^{-i\beta_jH_x}e^{-i\gamma_jH_z}e^{-i\pi H_z}\ket{+}^{\otimes n}\\
        & = \prod_{j=1}^p e^{-i\beta_jH_x} e^{-i\gamma_jH_z} \ket{+}^{\otimes n} \\
        & = \ket{\psi_p(\vgamma, \vbeta)}.
    \end{align}
    Also, by combining Eq.~(\ref{eqn:angle-reversal}) and (\ref{eqn:even-reg-period}), we can derive the symmetry.
    \begin{equation}
        F_p(\vgamma, \vbeta) = F_p(-\vgamma, -\vbeta) = F_p\left(\pi - \vgamma, \frac{\pi}{2} - \vbeta\right).
    \end{equation}
\end{proof}

From Eq.~(\ref{eqn:general-symmetry}), if the expectation $F_p$ has a global optimum at the point $(\vgamma^*, \vbeta^*)$,
then it will also have a global optimum at $(2\pi-\vgamma^*, \frac{\pi}{2}-\vbeta^*)$.
Therefore, to avoid redundancies in general graphs, it is suitable to set the optimization bounds as $\bm{\Phi}_p \in [0, \pi)^p\times [0, \pi/2)^p$.
On the other hand, for even-regular graphs, Eq.~(\ref{eqn:even-reg-symmetry}) shows that an extra symmetry exists at $(\pi-\vgamma^*, \frac{\pi}{2}-\vbeta^*)$ due to the shortened period,
so the suitable bounds are $\bm{\Phi}_p \in [0, \pi/2)^p\times [0, \pi/2)^p$.
Hence, the choices of bounds in Table~\ref{tab:good-bounds} are justified.

\section{Non-adiabatic path for odd-regular graphs}\label{sec:non-adiabatic}
It is known from the previous section, the QAOA for Max-cut has symmetric properties in its expectation function.
In other words, we know that, generally, if the expectation $F_p$ has a global optimum at the point $(\vgamma^*, \vbeta^*)$,
then it will also have a global optimum at $(2\pi-\vgamma^*, \frac{\pi}{2}-\vbeta^*)$. 
For even-regular graphs, they have an extra symmetry at $(\pi-\vgamma^*, \frac{\pi}{2}-\vbeta^*)$ due to the shortened period. 
Understanding this allows us to predict, for example, the other symmetric optimum will also have the adiabatic path pattern shown in Fig.~\ref{fig:params-var}(a),
except that $\gamma_j$ will decrease and $\beta_j$ will increase, as the original optimum and the symmetric optimum are negatively related.

However, things are a bit different in odd-regular graphs. 
We found out that in odd-regular graphs, the other ``symmetric'' optimum follows the pattern shown in Fig.~\ref{fig:non-adiabatic-start}(b), with smooth decrease of $\gamma_j$ and oscillating $\beta_j$. 
In this section, we will derive the symmetry for odd-regular graphs.

\begin{theorem}[Symmetry for odd-regular graphs]
    For odd-regular graphs, the expectation has a symmetry that follows
    \begin{equation}
        F_p(\vgamma, \vbeta) = F_p(\pi-\vgamma, \tilde{\vbeta}),
        \label{eqn:odd-reg-symmetry}
    \end{equation}
    where $\tilde{\vbeta}$ has elements $\beta_j$ if $j$ is odd and $(\pi/2-\beta_j)$ if $j$ is even:
    \begin{equation}
        \tilde{\vbeta} \equiv
        \begin{pmatrix}
            \beta_1 \\
            \frac{\pi}{2} - \beta_2 \\
            \beta_3 \\
            \frac{\pi}{2} - \beta_4 \\
            \vdots
        \end{pmatrix}.
        \label{eqn:beta-tilde}
    \end{equation}
\end{theorem}

\begin{proof}
    For the $H_z$ of odd-regular graphs, $e^{-i\pi H_z} = Z^{\otimes n}$. Consider the ansatz with $\vgamma$ shifted by $\pi$:
    \begin{align}
        \ket{\psi_p(\vgamma + \pi, \vbeta)} & = \prod_{j=1}^p e^{-i\beta_j H_x} e^{-i\pi H_z} e^{-i\gamma_j H_z} \ket{+}^{\otimes n} \\
        & = \prod_{j=1}^p e^{-i\beta_j H_x} Z^{\otimes n} e^{-i\gamma_j H_z} \ket{+}^{\otimes n}.
    \end{align}
    The expectation is thus
    \begin{multline}
        F_p(\vgamma + \pi, \vbeta) \\
        = \bra{+}^{\otimes n} (e^{i\gamma_1 H_z} Z^{\otimes n} e^{i\beta_1 H_x} \cdots e^{i\gamma_p H_z} Z^{\otimes n} e^{i\beta_p H_x}) H_z \\
        (e^{-i\beta_p H_x} Z^{\otimes n} e^{-i\gamma_p H_z} \cdots e^{-i\beta_1 H_x} Z^{\otimes n} e^{-i\gamma_1 H_z}) \ket{+}^{\otimes n}.
        \label{eqn:odd-reg-shifted}
    \end{multline}
    We know that $Z^{\otimes n}$ commutes with $e^{-i\gamma H_z}$, but does not commute with $e^{-i\beta H_x}$. 
    However, both $Z^{\otimes n}$ and $e^{-i\beta H_x}$ are Hermitian, so
    \begin{align}
        Z^{\otimes n}e^{-i\beta H_x} & = (Z^{\otimes n}e^{-i\beta H_x})^\dagger \\
        & = e^{i\beta H_x}Z^{\otimes n}. \label{eqn:hermitian}
    \end{align}
    As we can see from the above, as $Z^{\otimes n}$ moves through $e^{-i\beta H_x}$, the sign of $\beta$ changes. 
    The goal here is to use this property to move the $Z$'s in Eq.~(\ref{eqn:odd-reg-shifted}) towards the center $H_z$, 
    so that the $Z$'s from the left and right cancels out at the center, since they commute with $H_z$.
    As $Z^{\otimes n}$ moves through the equation, the sign of $\beta$'s that the operator passes through will be toggled.
    It is not difficult to see that after all the $Z$'s have been canceled out, the resulting vector of $\vbeta$ will have elements with alternating signs.
    We define the resulting vector $\vbeta'$:
    \begin{equation}
        \vbeta' \equiv 
        \begin{pmatrix}
            -\beta_1 \\
            \beta_2 \\
            -\beta_3 \\
            \beta_4 \\
            \vdots
        \end{pmatrix}.
    \end{equation}
    Hence, we have 
    \begin{equation}
        F_p(\vgamma+\pi, \vbeta) = F_p(\vgamma, \vbeta').
    \end{equation}
    Shifting $\vgamma$ on both sides of the equation by $-\pi$ and applying the angle-reversal on the RHS yields
    \begin{align}
        F_p(\vgamma, \vbeta) & = F_p(\vgamma-\pi, \vbeta') \\
        & = F_p(\pi-\vgamma, -\vbeta').
    \end{align}
    To tidy up this, we want to get the values of $-\vbeta'$ in the range of $[0, \pi/2)^p$. 
    As mentioned in Remark~\ref{remark:any-shift}, the expectation is left unchanged if we shift any number of the parameter(s) by its period. 
    Thus, we shift the negative elements (even indices) in $-\vbeta'$ by $\pi/2$, so if $\beta_j\in [0, \pi/2)$, then $(\pi/2-\beta_j)$ will also be within the range.
    We use a new symbol $\tilde{\vbeta}$ to denote this parameter vector, arriving at Eq.~(\ref{eqn:beta-tilde}).
    Hence,
    \begin{equation}
        F_p(\vgamma, \vbeta) = F_p(\pi-\vgamma, -\vbeta') = F_p(\pi-\vgamma, \tilde{\vbeta}).
    \end{equation}
\end{proof}

Thus, Eq.~(\ref{eqn:odd-reg-symmetry}) and (\ref{eqn:beta-tilde}) explain the oscillation of $\beta$ values shown in Fig.~\ref{fig:non-adiabatic-start}(b).
By restricting the optimization bounds to $\bm{\Phi}_p \in [0, \pi/2)^p \times [0, \pi/2)^p$, the non-adiabatic path for odd-regular graphs can be avoided.

\bibliography{qaoaref}

\end{document}